\documentclass[journal,final, onecolumn]{IEEEtran}

\usepackage{graphicx}

\usepackage{cite}

\ifCLASSINFOpdf

\else

\fi

\usepackage[cmex10]{amsmath}
\usepackage{amssymb}
\usepackage{amsthm}
\usepackage{subeqnarray}
\DeclareMathOperator{\tr}{tr}
\DeclareMathOperator{\cov}{cov}

\allowdisplaybreaks

\begin{document}

\sloppy
%
\title{Secret Key Generation from Vector Gaussian Sources with Public and Private Communications}


\author{Yinfei~Xu,~\IEEEmembership{Member,~IEEE,}
  and~Daming~Cao
\thanks{Y. Xu is with the School of Information Science and Engineering, Southeast University, Nanjing, 210096, China (email: yinfeixu@seu.edu.cn).}
\thanks{D. Cao is with the Department of Computer Science, National University of Singapore, Singapore (email: dcscaod@nus.edu.sg).}
\thanks{This paper was presented in part at the 2017 9th International Conference on Wireless Communications and Signal Processing \cite{CK17}, and is accepted in part by the 2020 IEEE International Symposium on Information Theory.}
}

\maketitle

\begin{abstract}
In this paper, we consider the problem of secret key generation with one-way communication through both a rate-limited public channel and a rate-limited secure channels where the public channel is from Alice to Bob and Eve and the secure channel is from Alice to Bob. In this model, we do not pose any constraints on the sources, i.e. Bob is not degraded to or less noisy than Eve.  We obtain the optimal secret key rate in this problem, both for the discrete memoryless sources and vector Gaussian sources. The vector Gaussian characterization is derived by suitably applying the enhancement argument, and Proving a new extremal inequality. The extremal inequality can be seen as coupling of two extremal inequalities, which are related to the degraded compound MIMO Gaussian broadcast channel, and the vector generalization of Costa's entropy power inequality, accordingly.
\end{abstract}


\newtheorem{theorem}{Theorem}
\newtheorem{lemma}{Lemma}
\newtheorem{definition}{Definition}
\newtheorem{remark}{Remark}
\newtheorem{example}{Example}

%

\section{Introduction}

The problem of secret key generation was introduced by Ahlswede and Csisz{\'a}r \cite{ahlswede1993common}, and by Maurer \cite{maurer1993secret}, where two separate terminals, named Alice and Bob, observe the outcomes of a pair of correlated sources separately and want to generate a common secret key, which is concealed from an eavesdropper Eve, given that the terminals can communicate through a noiseless public channel which the eavesdropper has complete access to. In \cite{ahlswede1993common}, the secret key capacity of correlated sources was characterized when Alice and Bob are allowed to communicate once over a channel with unlimited capacity. The secrecy key capacity was found by Csisz{\'a}r and Narayan \cite{csiszar2000common} when there is a constraint on the rate of the public channel.\par

From then on, the secret key generation from correlated sources problem has been attracted considerable attention, both in limited rate constraints setting\cite{GA10,CDS12,CB14,TBS17,LCV17}, and unlimited rate constraints setting \cite{NN10,T13,ZLLW14,CZ14,CMZ18,CY19}. However, for many models of interest in practise, the key capacity problem still remains unsolved. In \cite{WO11}, the secret key generation problem through rate limited noiseless public channel was extended to source with continuous alphabets. The fundamental limits for vector Gaussian sources, which are natural models of multiple input multiple output (MIMO) systems, was characterized. In \cite{LCV}, a water filling solution was further derived for the product vector Gaussian sources.\par

In this paper, we consider the problem of secret key generation with one-way communication, where in addition to the rate limited public channel, which can be observed by both Bob and Eve, we add a secure channel, which only connects Alice and Bob. One of the motivations of this problem comes from wireless sensor network with fading channels, where the nodes want to share a secret key to encrypt their communication. In this scenario, the frequency selectivity of the fading channels will create both public and secure channels. More specifically, in some frequency bands, the links from Alice to both Bob and Eve are of good qualities, which constitute the public channel. In some other frequency bands, the link from Alice to Bob is of good quality, but the link from Alice to Eve is basically broken. These frequency bands can be viewed as a secure channel. Another motivating example comes from \cite{thai2014physical}, where Alice and Bob are nodes equipped with multiple antennas and they want to communicate to share a secret key with the help of multiple single-antenna relays employing the amplify-and-forward strategy. We assume that some relays are ``nice but curious''\cite{lima2007random} which can be viewed as Eve, while the other relays are simply nice. Therefore, the links through the curious relays are public, while the links through the nice relays are secure.\par

Our problem can be viewed as a special case of the problem of secret key generation from correlated sources with the broadcast channel introduced in \cite{prabhakaran2008secrecy}, where Alice, Bob and Eve are connected by a one-way broadcast channel and they observe the outcomes of correlated sources separately which are independent of the channel. This problem in general is very difficult, because it is hard to identify the optimal strategies to combine the two resources, the channel and the sources, to generate a secret key between Alice and Bob. Therefore, the secret key capacity in this problem remains unknown in its general form. Achievability schemes  and converses have been proposed in \cite{prabhakaran2012secrecy,bassi2016secretArXiv}. However, these achievabilities and converses in general do not match.\par

For the vector Gaussian sources problem, one of difficulties to show the fundamental tradeoff between the key capacity and the communication constraints is that vector Gaussian sources are not in general degraded. This difficulty frequently appears in several vector Gaussian multi-terminal problem \cite{MIMO,exinq,Liu09,Liu10,EU13}. In \cite{WO11}, Watanabe and Oohama circumvent this difficulty by suitably applying the enhancement argument. Further invoking the so-called Costa's type extremal inequality \cite{Liu10}, they showed that one Gaussian auxiliary random variable suffices to characterise the rate region. However, for the private and public communication available in our setting, it will been seen that single auxiliary random variable fails to characterize the tradeoff between the key capacity and communication rate. As a consequence, applying Costa's type extremal inequality alone is not sufficient when the public communication constraint is considered. The desired converse result can be eventually obtained by a suitable integration of the classical enhancement argument. With the enhanced source model, the corresponding extremal inequality should be decoupled into two enhanced extremal inequalities, in which one is related to the degraded compound MIMO Gaussian broadcast channel in \cite{WLSSV09,CL14}, and the other one is the vector generalization of Costa's entropy power inequality in \cite{Liu10,WO11}.


The rest of this paper is organized as follows. The problem setup is given in Section II, we first derive the optimal achievable rate region for the case of discrete memoryless sources case, and then we show the rate region characterization for the case of vector Gaussian sources considered in this paper. The achievablility and converse proof for the discrete memoryless sources are shown in Section III.
Section IV is devoted to proof our new extremal inequality, and we show the Gaussian auxiliary random variables suffice to achieve the optimal rate region, for the vector Gaussian sources. In Section V, we conclude with a summary of our results and a remark on future research.

\section{Problem Statement and Main Result}

\subsection{Discrete Memoryless Sources}
Consider a network with three nodes, including a transmitter Alice, a receiver Bob and an eavesdropper Eve. We assume three discrete memoryless sources indicated by random variables $(X,Y,Z)$, defined in the alphabets $(\mathcal{X},\mathcal{Y},\mathcal{Z})$, respectively.  We assume that Alice and Bob observe the $n$-length source sequences $X^n$ and $Y^n$, respectively, and Eve observes $n$-length source sequence $Z^n$. In order to generate a secret key $K$, which is shared by Alice and Bob and concealed from Eve, Alice can send two messages $M_1$ and $M_2$, where $M_1$ is through a noiseless public channel, which can be observed by Eve, and $M_2$ is through a noiseless secure channel, to which Eve has no access.

A $(2^{nR_{1}}, 2^{nR_{2}}, n)$ code consists of
\begin{itemize}
 \item a public encoding function $\phi_{1}: \mathcal{X}^{ n} \mapsto \mathcal{M}_{1}^{n} = [1: 2^{nR_1}]$ that finds a codeword $m_{1}(x^{n})$ to each $n$-length source sequence $x^{n}$, and sends it to both Bob and Eve,
 \item a private encoding function $\phi_{2}: \mathcal{X}^{ n} \mapsto \mathcal{M}_{2}^{n} = [1:2^{nR_2}]$ that finds a codeword $m_{2}(x^{n})$ to each $n$-length source sequence $x^{n}$, and sends it to both Bob only,
  \item a key generation function $\psi_{1}: \mathcal{X}^{n} \mapsto \mathcal{K}^{n} = [1: 2^{nR_K}]$ that assigns a random mapping $k_{1}(x^{n})$ by giving Alice's $n$-length source sequence $x^{n}$,
  \item a key generation function $\psi_{2}: \mathcal{Y}^{n} \times \mathcal{M}_{1}^{n} \times \mathcal{M}_{2}^{n} \mapsto \mathcal{K}^{n}= [1: 2^{nR_K}]$ that assigns a random mapping $k_{2}(y^{n}, m_{1}, m_{2})$ by giving Bob's $n$-length source sequence $y^{n}$ and all received indices $m_{1}$ and $m_{2}$.
\end{itemize}

Then the secret key is generated by Alice and Bob from the functions $\psi_{1}$ and $\psi_{2}$, respectively, which should agree with probability $1$ and be concealed from Eve. The probability of error for the key generation code is defined as
\begin{equation}
P_e^{(n)}=\Pr\{  K_{1} \neq K_{2}      \}.
\end{equation}
The \emph{key leakage rate} at Eve is defined as
\begin{equation}
R_L^{(n)}=\max_{j \in \{1,2\}}\frac{1}{n}I(K_{j};Z^n, M_1).
\end{equation}

\begin{definition}
A secret key rate $R_K$ with constraint communication rate pair $(R_1,R_2)$ is achievable if there exists a sequence of $(2^{nR_{1}}, 2^{nR_{2}}, n)$ code such that
\begin{align}
\lim_{n \rightarrow \infty} P_e^{(n)} =0,\\
\lim_{n \rightarrow \infty} R_{L}^{(n)}=0.\label{SMSC1}
\end{align}
\end{definition}

For the discrete memoryless source setting, we have the following single-letter expression on the largest achievable secret key rate $R_{K}$ with public and private communication constraints $R_{1}$ and $R_{2}$.
\begin{theorem}\label{mainth}
For the discrete memoryless source secret key generation problem with public and private communication constraints, the rate tuple $(R_K,R_1,R_2)$ is achievable if and only if
\begin{align}
   R_K-R_{2} & \leq I(U;Y|V)-I(U;Z|V),\label{TH1}\\
  R_1+R_2 &\ge I(U;X|Y), \label{TH2}  \\
  R_1&\ge I(V;X|Y), \label{TH11}
\end{align}
where random variables $(V,U,X,Y,Z)$ satisfy the following Markov chain
\begin{align}
  V \rightarrow U &\rightarrow X \rightarrow (Y,Z).
\end{align}
\end{theorem}
\begin{proof}
See Section \ref{prf1}.
\end{proof}

\subsection{Vector Gaussian Sources}

\begin{figure}
	\centering
	\setlength{\unitlength}{0.65cm}
	\begin{picture}(21,8)
	\linethickness{1pt}
	\put(2,5){\framebox(3.5,2.5)}
	\put(15,5){\framebox(3.5,2.5)}
	\put(8.5,1){\framebox(4.5,2)}
	\put(2.8,6.1){\makebox{Encoder}}
	\put(15.9,6.1){\makebox{Decoder}}
	\put(9.3,1.9){\makebox{Eavesdropper}}
	\put(0,6.25){\vector(1,0){2}}
	\put(20.5,6.25){\vector(-1,0){2}}
	\put(5.5,7){\vector(1,0){9.5}}
	\put(5.5,5.5){\vector(1,0){9.5}}
	\put(6.5,2){\vector(1,0){2}}
    \put(3.5,5){\vector(0,-1){1.5}}
    \put(16.7,5){\vector(0,-1){1.5}}
	\put(10.25,5.5){\vector(0,-1){2.5}}
	\put(1,6.75){\makebox(0,0){\large $\mathbf{X}^n$}}
	\put(20,6.75){\makebox(0,0){\large $\mathbf{Y}^n$}}
	\put(7.5,2.5){\makebox(0,0){\large $\mathbf{Z}^n$}}
	\put(10,7.6){\makebox(0,0){\large $M_2$}}
	\put(10,6.1){\makebox(0,0){\large $M_1$}}
     \put(16.8,2.8){\makebox(0,0){\large $K_2$}}
     \put(3.6,2.8){\makebox(0,0){\large $K_1$}}
	\end{picture}
	\caption{Secret key generation with rate constrained public and private communications}\label{model}
\end{figure}
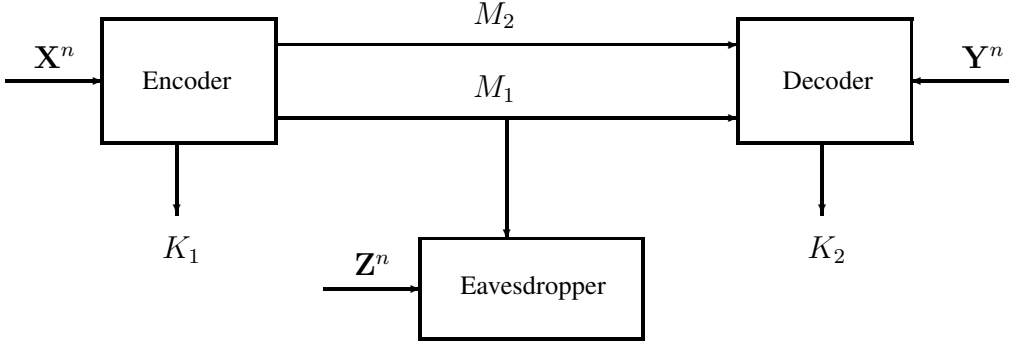

Now we study the same communication constrained secret key generation problem, for the vector Gaussian sources setting (see Fig.\ref{model}). Let $\{\mathbf{X}(t), \mathbf{Y}(t), \mathbf{Z}(t)\}_{t=1}^{n}$ be i.i.d. vector-valued discrete time Gaussian sources, where across the time index $t$, each tuple is drawn from the same jointly vector Gaussian distribution. The encoder, the legitimate decoder and the eavesdropper decoder observe $\mathbf{X}(t)$, $\mathbf{Y}(t)$ and $\mathbf{Z}(t)$, respectively. The vector Gaussian source $\{\mathbf{X}(t), \mathbf{Y}(t), \mathbf{Z}(t)\}_{t=1}^{n }$ can be written as
\begin{align}
\mathbf{Y}(t) =  \mathbf{X} (t) + \mathbf{N}_{Y}(t), \label{eq:source1}\\
\mathbf{Z}(t) =  \mathbf{X} (t) + \mathbf{N}_{Z}(t), \label{eq:source2}
\end{align}
where each $\mathbf{X}(t)$ is a $p\times1$-dimensional Gaussian random vector with mean zero and covariance $\mathbf{K} \succ0$, each $\mathbf{N}_{Y}(t)$ is a $p\times1$-dimensional Gaussian random vector with mean zero and covariance $\mathbf{K}_{Y} \succ 0 $, and $\mathbf{N}_{Z}(t)$ is a $p\times1$-dimensional Gaussian random vector with mean zero and covariance $\mathbf{K}_{Z} \succ 0 $, respectively. We shall point out that $(\mathbf{N}_{Y}(t), \mathbf{N}_{Z}(t))$ and $\mathbf{X}(t)$ are independent from expressions \eqref{eq:source1} and \eqref{eq:source2}. However, no additional independence relationship is imposed between $\mathbf{N}_{Y}(t)$ and $\mathbf{N}_{Z}(t)$.

In \cite{WO11}, the authors showed that a single layer code suffices, and characterize the optimal trade-off on $R_{K}$ and $R_{1}$ for the vector Gaussian sources setting. Their converse method is motivated by the \emph{enhancement} argument \cite{MIMO}, \cite{exinq} for vector Gaussian wiretap channel in \cite{Liu09}. In this paper, we show a similar enhancement on two-layer superposition codes can be applied to establish the converse proof on the optimal $(R_{K}, R_{1}, R_{2})$ trade-off problem.

According to Theorem \ref{mainth} for discrete memoryless sources , a single-letter description of the optimal $(R_{K}, R_{1}, R_{2})$ trade-off for the vector Gaussian secret key generation problem can be given as follows.

\begin{theorem}\label{mainth2}
For the vector Gaussian secret key generation problem with public and private communication constraints, the rate tuple $(R_K,R_1,R_2)$ is achievable if and only if
\begin{align}
   R_K-R_{2} & \leq \frac{1}{2} \log \frac{|\mathbf{K}+\mathbf{K}_{Y}-\mathbf{B}_{1}|}{|\mathbf{K}+\mathbf{K}_{Y}-\mathbf{B}_{1}-\mathbf{B}_{2}|} - \frac{1}{2} \log \frac{|\mathbf{K}+\mathbf{K}_{Z}-\mathbf{B}_{1}|}{|\mathbf{K}+\mathbf{K}_{Z}-\mathbf{B}_{1}-\mathbf{B}_{2}|},    \label{eq:alin1}   \\
  R_1+R_2 &\ge \frac{1}{2} \log \frac{|\mathbf{K}|}{|\mathbf{K}-\mathbf{B}_{1}-\mathbf{B}_{2}|}-\frac{1}{2}\log \frac{|\mathbf{K}+\mathbf{K}_{Y}|}{|\mathbf{K}+\mathbf{K}_{Y}-\mathbf{B}_{1}-\mathbf{B}_{2}|}, \label{eq:align2}\\
  R_1&\ge \frac{1}{2} \log \frac{|\mathbf{K}|}{|\mathbf{K}-\mathbf{B}_{1}|}-\frac{1}{2}\log \frac{|\mathbf{K}+\mathbf{K}_{Y}|}{|\mathbf{K}+\mathbf{K}_{Y}-\mathbf{B}_{1}|}, \label{eq:align3}
\end{align}
for some positive semi-definite matrices $\mathbf{B}_{1}$, $\mathbf{B}_{2} \succeq \mathbf{0}$.
\end{theorem}
\begin{proof}
The achievable part of Theorem \ref{mainth2} is based on constructing Gaussian test channels to maintain the Markov chain $ V \rightarrow U \rightarrow \mathbf{X} \rightarrow (\mathbf{Y},\mathbf{Z})$, and the details can be found in Appendix \ref{app:0}. The converse part of Theorem \ref{mainth2} relies on a careful combination of the usage of channel enhancement argument and extremal inequalities, and the details can be found in Section \ref{prf2}.
\end{proof}
\begin{remark}
It should be pointed out that the jointly vector Gaussian sources given by \eqref{eq:source1} and \eqref{eq:source2} are not in the most general form. In the general case, we have
\begin{align}
\mathbf{Y}(t) =  \mathbf{H}_{Y}\mathbf{X} (t) + \mathbf{N}_{Y}(t), \\
\mathbf{Z}(t) =  \mathbf{H}_{Z}\mathbf{X} (t) + \mathbf{N}_{Z}(t),
\end{align}
for some matrices $\mathbf{H}_{Y} \in \mathbb{R}^{m_{Y} \times p}$, $\mathbf{H}_{Z} \in \mathbb{R}^{m_{Z} \times p}$. The extension of Theorem \ref{mainth2} can be obtained by following the lines of \cite[Sec. V]{WO11}, in a similar manner.
\end{remark}

\section{Proof of Theorem \ref{mainth}}\label{prf1}

\subsection{Principles of the Achievability}
The achievability scheme that we propose in this paper can be viewed as a combination of the scheme in secret key generation from correlated sources \cite{csiszar2000common}, consisting of a codebook of the superposition structure, and secret key distribution through secret channel. Moreover, this scheme can be viewed as a special case of the separated achievable scheme in \cite[Th.2]{bassi2016secretArXiv}.\par

The rate of the public channel and the secure channel are used for the transmission of the following:
\begin{enumerate}
\item the inner code $V^n$.
\item the outer code $U^n$.
\item key distribution.
\end{enumerate}
Our proposed scheme follows the principles below:
\begin{itemize}
\item The public channel is used to transmit the inner code $V^n$. Since the rate of the inner code is less than the rate of the public channel, then the leftover rate of the public channel will be used to transmit the outer code $U^n$.
\item  The secure channel is used to transmit the outer code $U^n$. Since there is still extra rate leftover in the secure channel, then the leftover rate of the secure channel can be used for key distribution.
\item The public channel can not be used for key distribution and the secure channel can not be used to transmit the inner code $V^n$.
\end{itemize}

We show that the proposed scheme achieves the rate in Theorem \ref{mainth} and is thus, optimal. For the completeness, the details of the proof can be found in Appendix \ref{app:1}.

\subsection{The Converse}

We begin the proof of the converse with
\begin{align}
&nR_K\leq  H(K) \\
  &\overset{(a)}\leq H(K)-H(K|Y^n,M_1,M_2)+n\epsilon \label{ConK1}\\
   &=I(K;Y^n,M_1,M_2)+n\epsilon \\
   &\overset{(b)}\leq  I(K;Y^n,M_1,M_2)-I(K;Z^n,M_1)+2n\epsilon \label{ConK11}\\
   &=I(K;Y^n,M_1)+I(K;M_2|Y^n,M_1)-I(K,M_2;Z^n,M_1)-I(M_2;Z^n,M_1|K)+2n\epsilon \\
   &=I(K,M_2;Y^n,M_1)+I(K;M_2|Y^n,M_1)-I(M_2;Y^n,M_1|K)-I(K,M_2;Z^n,M_1)\nonumber\\
   &\quad-I(M_2;Z^n,M_1|K)+2n\epsilon \\
   &=I(K,M_2;Y^n,M_1)+I(K;M_2|Y^n,M_1)-I(M_2;Y^n,M_1,K)-I(K,M_2;Z^n,M_1)\nonumber\\
   &\quad-I(M_2;Z^n,M_1,K)+2n\epsilon \\
   &=I(K,M_2;Y^n,M_1)-I(M_2;Y^n,M_1)-I(K,M_2;Z^n,M_1)-I(M_2;Z^n,M_1,K)+2n\epsilon \\
   &=I(K,M_2;Y^n|M_1)-I(K,M_2;Z^n|M_1)-I(M_2;Y^n,M_1)-I(M_2;Z^n,M_1,K)+2n\epsilon \\
   &= I(K,M_2;Y^n|M_1)-I(K,M_2;Z^n|M_1)+H(M_2)-I(M_2;Y^n,M_1)-H(M_2|Z^n,M_1,K)+2n\epsilon \\
   &\overset{(c)}\leq I(K,M_2;Y^n|M_1)-I(K,M_2;Z^n|M_1)+nR_2+2n\epsilon.
\end{align}
where
\begin{enumerate}
\item[(a)] follows from Fano's inequality;
\item[(b)] follows from the secrecy constraint in (\ref{SMSC1});
\item[(c)]follows by the non-negativity of discrete entropy and mutual information.
\end{enumerate}

By applying the key identity \cite[Lemma 17.12]{csiszar2011information}, it follows that
\begin{align}
  I(K,M_2;Y^n|M_1)&-I(K,M_2;Z^n|M_1)= n[I(U;Y_J|V)-I(U;Z_J|V)],
\end{align}
where
\begin{align}
V&\triangleq (M_1,Y^{i-1},Z_{i+1}^n,J),\label{Tdef}\\
U&\triangleq (K,M_2,T).\label{Udef}
\end{align}
Here, $J$ is uniformly distributed on $[1:n]$ and independent of $(X^n,Y^n,Z^n)$. Since $K$, $M_1$ and $M_2$ is a function of $X^n$, the Markov Chain $V \rightarrow U \rightarrow X \rightarrow (Y,Z) $ is satisfied. Because of the fact that $(X^n,Y^n,Z^n)$ are i.i.d., we can replace $Y_J$ and $Z_J$ by $Y$ and $Z$. Then (\ref{TH1}) is proved.

Next, we consider the sum rate $(R_1+R_2)$
\begin{align}
 &n(R_1+R_2)\geq H(M_1,M_2)\\
   &\geq H(M_1,M_2|Y^n) \\
   &\overset{(a)}\geq H(K,M_1,M_2|Y^n)-n\epsilon \label{ConSR1}\\
   &\geq H(K,M_1,M_2|Y^n)-H(K,M_1,M_2|X^n) -n\epsilon\\
   &= I(K,M_1,M_2;X^n)-I(K,M_1,M_2;Y^n)-n\epsilon\nonumber\\
   &\overset{(b)}= \sum_{i=1}^n[I(K,M_1,M_2;X_i|X_{i+1}^n,Y^{i-1})-I(K,M_1,M_2;Y_i|X_{i+1}^n,Y^{i-1})]-n\epsilon \label{ConSR2}\\
   &\overset{(c)}= \sum_{i=1}^n[I(K,M_1,M_2,X_{i+1}^n,Y^{i-1};X_i)-I(K,M_1,M_2,X_{i+1}^n,Y^{i-1};Y_i)] -n\epsilon\label{ConSR3}\\
   &\overset{(d)}= \sum_{i=1}^n[I(K,M_1,M_2,X_{i+1}^n,Y^{i-1},Z_{i+1}^n;X_i)-I(K,M_1,M_2,X_{i+1}^n,Y^{i-1},Z_{i+1}^n;Y_i)] -n\epsilon\label{ConSR4}\\
   &\overset{(e)}= \sum_{i=1}^n[I(K,M_1,M_2,X_{i+1}^n,Y^{i-1},Z_{i+1}^n;X_i|Y_i)-n\epsilon\label{ConSR41}\\
   &\overset{(f)}\geq \sum_{i=1}^n[I(K,M_1,M_2,Y^{i-1},Z_{i+1}^n;X_i)-I(K,M_1,M_2,Y^{i-1},Z_{i+1}^n;Y_i)]-n\epsilon,\label{ConSR5}
\end{align}
where
\begin{enumerate}
\item[(a)] follows from Fano's inequality;
\item[(b)] follows from the key identity \cite[Lemma 17.12]{csiszar2011information};
\item[(c)] follows because $(X^n,Y^n,Z^n)$ are i.i.d.;
\item[(d)] follows from the Markov Chain $Z_{i+1}^n \rightarrow (K,M_1,M_2,X_{i+1}^n,Y^{i-1}) \rightarrow (X_i,Y_i) $;
\item[(e)] follows from the Markov Chain $Y_i\rightarrow X_i\rightarrow (K,M_1,M_2,X_{i+1}^n,Y^{i-1},Z_{i+1}^n)$;
\item[(f)] follows from the Markov Chain $X_{i+1}^n \rightarrow (K,M_1,M_2,Y^{i-1},Z_{i+1}^n,X_i) \rightarrow Y_i $.
\end{enumerate}
We thereby have
\begin{align}
  & n(R_1+R_2) \geq  n[I(U;X_J)-I(U;Y_J)] \\
   &= nI(U;X_J|Y_J)\\
   &= nI(U;X|Y).
\end{align}
with $U$ as defined in (\ref{Udef}) and $J$ is uniform on $[1:n]$.

Finally, for the public rate $R_1$, we have
\begin{align}
  & nR_1 \geq H(M_1)\\
  &\geq H(M_1|Y^n) \\
   &\geq H(M_1|Y^n)-H(M_1|X^n) \\
   &= I(M_1;X^n)-I(M_1;Y^n).
\end{align}
By the similar argument as in the sum rate derivation, we can get 
\begin{align}
  & nR_1 \geq \sum_{i=1}^n[I(M_1,Y^{i-1},Z_{i+1}^n;X_i)-I(M_1,Y^{i-1},Z_{i+1}^n;Y_i)]  \\
   &\geq  n[I(V;X_J)-I(V;Y_J)] \\
   &= nI(V;X_J|Y_J)\\
   &= nI(V;X|Y).
\end{align}
where $V$ as defined in (\ref{Tdef}) and $J$ is uniform on $[1:n]$,  which concludes the converse proof of Theorem \ref{mainth}.

\section{The Converse of Theorem \ref{mainth2}}\label{prf2}

\subsection{The Extremal Inequality}

As in \cite{WO11}, the achievable rate region of $(R_{K}, R_{1}, R_{2})$ for vector Gaussian sources is defined as
\begin{equation}
\mathcal{R}(\mathbf{X}, \mathbf{Y}, \mathbf{Z}) \triangleq \left\{   (R_{K}, R_{1}, R_{2}): (R_{K}, R_{1}, R_{2}) \text{ is achievable}       \right\}.
\end{equation}
Due to the convexity of $\mathcal{R}(\mathbf{X}, \mathbf{Y}, \mathbf{Z})$, to characterize the optimal trade-off of $(R_{K}, R_{1}, R_{2})$  for the vector Gaussian model, we can write the following $\mu$-sum problem, alternatively,
\begin{align}
&\inf_{(R_{K}, R_{1}, R_{2}) \in \mathcal{R}(\mathbf{X}, \mathbf{Y}, \mathbf{Z}) } \mu_{1} (R_{2} - R_{K}) + \mu_{2} (R_{1}+R_{2}) + \mu_{3} R_{1} \\
&=\inf_{(R_{K}, R_{1}, R_{2}) \in \mathcal{R}(\mathbf{X}, \mathbf{Y}, \mathbf{Z}) } (\mu_{2}+ \mu_{3})R_{1} + (\mu_{1}+\mu_{2})R_{2} - \mu_{1} R_{K},
\end{align}
for any $\mu_{1}, \mu_{2}, \mu_{3} \geq 0$. To prove the converse part of Theorem \ref{mainth2}, it is equivalent to show the following inequality folds for any $(R_{K}, R_{1}, R_{2}) \in \mathcal{R}(\mathbf{X}, \mathbf{Y}, \mathbf{Z})$,
\begin{equation}
(\mu_{2}+ \mu_{3})R_{1} + (\mu_{1}+\mu_{2})R_{2} - \mu_{1} R_{K} \geq \mathfrak{R}^{*}(\mu_{1}, \mu_{2}, \mu_{3}),
\end{equation}
where $\mathfrak{R}^{*}(\mu_{1}, \mu_{2}, \mu_{3})$ is a Gaussian optimization problem shown as follows
\begin{align} \label{eqn:opt}
& \mathfrak{R}^{*}(\mu_{1}, \mu_{2}, \mu_{3}) \nonumber \\
&\triangleq \min_{\mathbf{B}_{1},\mathbf{B}_{2}}\frac{\mu_1+\mu_2}{2} \log \left|  \mathbf{K} + \mathbf{K}_{Y} - \mathbf{B}_{1}-\mathbf{B}_{2} \right|- \frac{\mu_1}{2} \log \left|\mathbf{K} + \mathbf{K}_{Z}- \mathbf{B}_{1}-\mathbf{B}_{2} \right| \nonumber \\
&\qquad\quad- \frac{\mu_{2}}{2} \log \left|  \mathbf{K}  - \mathbf{B}_{1}- \mathbf{B}_{2} \right| + \frac{\mu_1}{2} \log \left|  \mathbf{K} + \mathbf{K}_{Z} - \mathbf{B}_{1} \right| \nonumber \\
&\qquad\quad+ \frac{\mu_3-\mu_1}{2} \log \left|  \mathbf{K} + \mathbf{K}_{Y} - \mathbf{B}_{1} \right| - \frac{\mu_3}{2} \log \left|  \mathbf{K}  - \mathbf{B}_{1} \right| \nonumber \\
&\qquad\quad+ \frac{\mu_2+\mu_3}{2} \log |\mathbf{K}| -\frac{\mu_2+\mu_3}{2} \log |\mathbf{K}+\mathbf{K}_{Y}|  \nonumber \\
&\qquad\text{subject to}\quad \mathbf{B}_{1} \succeq \mathbf{0}, \; \mathbf{B}_{2} \succeq \mathbf{0}.
\end{align}

Let $(\mathbf{B}_{1}^{*},\mathbf{B}_{2}^{*})$ be one (non-unique) minimizer of the optimization problem $\mathfrak{R}^{*}(\mu_{1}, \mu_{2}, \mu_{3})$. The necessary Karush-Kuhn-Tucker (KKT) conditions are given in the following lemma.
\begin{lemma} \label{lemma_KKT} The minimizer $(\mathbf{B}_{1}^{*},\mathbf{B}_{2}^{*})$ of $\mathfrak{R}^{*}(\mu_{1}, \mu_{2}, \mu_{3})$ need to satisfy
\begin{align}
 \frac{\mu_1}{2} \left(\mathbf{K} + \mathbf{K}_{Z}- \mathbf{B}_{1}^{*}-\mathbf{B}_{2}^{*} \right)^{-1}+ \frac{\mu_{2}}{2}  \left(  \mathbf{K}  - \mathbf{B}_{1}^{*}- \mathbf{B}_{2}^{*} \right)^{-1}=\frac{\mu_1+\mu_2}{2} \left(  \mathbf{K} + \mathbf{K}_{Y} - \mathbf{B}_{1}^{*}-\mathbf{B}_{2}^{*} \right)^{-1} + \mathbf{M}_{2}, \label{eq:KKT1}\\
   \frac{\mu_3}{2} \left( \mathbf{K}  - \mathbf{B}_{1}^{*} \right)^{-1} + \frac{\mu_1-\mu_3}{2}  \left(  \mathbf{K} + \mathbf{K}_{Y} - \mathbf{B}_{1}^{*} \right)^{-1} +\mathbf{M}_{2} =  \frac{\mu_1}{2}  \left(  \mathbf{K} + \mathbf{K}_{Z} - \mathbf{B}_{1}^{*} \right)^{-1} + \mathbf{M}_{1}, \label{eq:KKT2}
\end{align}
for some positive semi-definite matrices $\mathbf{B}_{1}^{*}, \mathbf{B}_{2}^{*}, \mathbf{M}_{1}, \mathbf{M}_{2} \succeq0$ such that
\begin{align}
\mathbf{B}_{1}^{*} \mathbf{M}_{1} = \mathbf{M}_{1}\mathbf{B}_{1}^{*} =\mathbf{0}, \label{eq:KKT4}\\
\mathbf{B}_{2}^{*} \mathbf{M}_{2} = \mathbf{M}_{2}\mathbf{B}_{2}^{*} =\mathbf{0}. \label{eq:KKT5}
\end{align}
\end{lemma}
\begin{proof}
See Appendix \ref{KKT}.
\end{proof}

Now Starting from the single-letter expressions in Theorem \ref{mainth}, the $\mu$-sum problem for any rate tuple in $\mathcal{R}(\mathbf{X}, \mathbf{Y}, \mathbf{Z})$ should be lower bounded by
\begin{align}
 &(\mu_{2}+ \mu_{3})R_{1} + (\mu_{1}+\mu_{2})R_{2} - \mu_{1} R_{K}\\
 &=\mu_{1} (R_{2} - R_{K}) + \mu_{2} (R_{1}+R_{2}) + \mu_{3} R_{1} \\
 & \geq \mu_{1} (I(U;\mathbf{Z}|V)-I(U;\mathbf{Y}|V)) + \mu_{2} I(U; \mathbf{X} | \mathbf{Y}) + \mu_3 I(V; \mathbf{X} | \mathbf{Y})\\
 & \overset{(a)} = \mu_{1} (I(U;\mathbf{Z}|V)-I(U;\mathbf{Y}|V)) + \mu_{2} ( I(U; \mathbf{X}) - I (U; \mathbf{Y})          ) + \mu_{3} ( I(V; \mathbf{X}) - I (V; \mathbf{Y})          )\\
 & \overset{(b)} = \mu_{1} h(\mathbf{Z}|V)    +       (\mu_3 - \mu_1) h(\mathbf{Y}|V) - \mu_3 h(\mathbf{X}|V )+(\mu_1+\mu_2)h(\mathbf{Y}|U) - \mu_2 h(\mathbf{X}|U) - \mu_{1} h(\mathbf{Z}|U) \nonumber \\
 & \quad\; + (\mu_2+\mu_3)h(\mathbf{X}) - (\mu_2+\mu_3)h(\mathbf{Y}) \\
 & = (\mu_1+\mu_2)h(\mathbf{Y}|U)- \mu_{1} h(\mathbf{Z}|U) - \mu_2 h(\mathbf{X}|U)+  \mu_{1} h(\mathbf{Z}|V)    +       (\mu_3 - \mu_1) h(\mathbf{Y}|V) - \mu_3 h(\mathbf{X}|V )  \nonumber \\
 & \quad \; + \frac{\mu_2+\mu_3}{2} \log |\mathbf{K}| -\frac{\mu_2+\mu_3}{2} \log |\mathbf{K}+\mathbf{K}_{Y}|. \label{eqn:lb}
\end{align}
where
\begin{enumerate}
\item[(a)] follows from Makov Chain $(U,V) \rightarrow \mathbf{X} \rightarrow \mathbf{Y}$,
\item[(b)]  follows from Makov Chain $V \rightarrow U \rightarrow \mathbf{X} \rightarrow (\mathbf{Y}, \mathbf{Z})$.
\end{enumerate}

By comparing \eqref{eqn:lb} with optimization problem $\mathfrak{R}^{*}(\mu_{1}, \mu_{2}, \mu_{3})$ in \eqref{eqn:opt}, it can be shown that to prove Theorem \ref{mainth2}, it is sufficient to prove the following extremal inequality

\begin{theorem} \label{ext_thm}
There exist two positive semi-definite matrices $\mathbf{B}_{1}^{*}$ and $\mathbf{B}_{2}^{*}$, which satisfy KKT conditions \eqref{eq:KKT1}-\eqref{eq:KKT5} in lemma \ref{lemma_KKT} to minimize optimization problem $\mathfrak{R}^{*}(\mu_{1}, \mu_{2}, \mu_{3})$, then for some real numbers $\mu_1, \mu_2, \mu_3 \geq 0$, we have
\begin{align}
&(\mu_1+\mu_2)h(\mathbf{Y}|U)- \mu_{1} h(\mathbf{Z}|U) - \mu_2 h(\mathbf{X}|U)\nonumber \\
&+  \mu_{1} h(\mathbf{Z}|V)    +       (\mu_3 - \mu_1) h(\mathbf{Y}|V) - \mu_3 h(\mathbf{X}|V ) \nonumber \\
&\geq \; \frac{\mu_1+\mu_2}{2} \log \left|  \mathbf{K} + \mathbf{K}_{Y} - \mathbf{B}_{1}^{*}-\mathbf{B}_{2}^{*} \right|- \frac{\mu_1}{2} \log \left|\mathbf{K} + \mathbf{K}_{Z}- \mathbf{B}_{1}^{*}-\mathbf{B}_{2}^{*} \right| \nonumber \\
&\quad- \frac{\mu_{2}}{2} \log \left|  \mathbf{K}  - \mathbf{B}_{1}^{*}- \mathbf{B}_{2}^{*} \right| + \frac{\mu_1}{2} \log \left|  \mathbf{K} + \mathbf{K}_{Z} - \mathbf{B}_{1}^{*} \right| \nonumber \\
&\quad+ \frac{\mu_3-\mu_1}{2} \log \left|  \mathbf{K} + \mathbf{K}_{Y} - \mathbf{B}_{1}^{*} \right| - \frac{\mu_3}{2} \log \left|  \mathbf{K}  - \mathbf{B}_{1}^{*} \right| , \label{eq:exinq}
\end{align}
for any $(U,V)$ such that $V \rightarrow U \rightarrow \mathbf{X} \rightarrow (\mathbf{Y}, \mathbf{Z})$ forms a Markov chain.
\end{theorem}

The proof of \eqref{eq:exinq} depends on the \emph{enhancement} argument introduced in\cite{MIMO}, \cite{exinq},  which can be divided into two steps.
In the first step, we enhance the source $\mathbf{Y}$ to $\tilde{\mathbf{Y}}$ such that the Markov chain $\mathbf{X} \rightarrow \tilde{\mathbf{Y}} \rightarrow (\mathbf{Y}, \mathbf{Z})$ holds. In the second step, we decouple to extremal inequality \eqref{eqn:lb} to two new ones, associated with enhanced sources $\tilde{\mathbf{Y}}$, respectively. The proof of extremal inequality \eqref{eqn:lb} in provided by involving the two enhanced extremal inequalities, which is related to the degraded compound MIMO Gaussian broadcast channel in \cite{WLSSV09,CL14}, and the vector generalization of Costa's entropy power inequality in \cite{Liu10,WO11}.

\subsection{Some Lemmas}

In order to reduce the non-degraded sources to the degraded case, we introduce a new covariance matrix such that
\begin{align}
\frac{\mu_1+\mu_2}{2}\left(\mathbf{K}+\tilde{\mathbf{K}}_{Y} - \mathbf{B}_{1}^{*} - \mathbf{B}_{2}^{*}\right)^{-1}& =\frac{\mu_1+\mu_2}{2}\left(\mathbf{K}  + \mathbf{K}_{Y}- \mathbf{B}_{1}^{*} - \mathbf{B}_{2}^{*}\right)^{-1} + \mathbf{M}_{2} \succ \mathbf{0}. \label{eq:def_Y}
\end{align}
Then $\tilde{\mathbf{K}}_{Y}$ has useful properties listed in the following lemma.
\begin{lemma} \label{lemma_enhancement}
$\tilde{\mathbf{K}}_{Y}$ has the following properties:

\begin{enumerate}
\item
$ \mathbf{0}    \prec     \tilde{\mathbf{K}}_{Y}  \preceq {\mathbf{K}}_{Y};$
\item
$ \tilde{\mathbf{K}}_{Y}  \preceq {\mathbf{K}}_{Z};$
\item
$
\frac{\mu_1+\mu_2}{2}\left(\mathbf{K}+\tilde{\mathbf{K}}_{Y} - \mathbf{B}_{1}^{*} \right)^{-1} =\frac{\mu_1+\mu_2}{2}\left(\mathbf{K}  + \mathbf{K}_{Y}- \mathbf{B}_{1}^{*} \right)^{-1} + \mathbf{M}_{2};
$
\item
$
\left(\mathbf{K} + \tilde{\mathbf{K}}_{Y}- \mathbf{B}_{1}^{*}- \mathbf{B}_{2}^{*}\right)^{-1}\left(\mathbf{K} + \tilde{\mathbf{K}}_{Y}- \mathbf{B}_{1}^{*}\right) = \left(\mathbf{K} + {\mathbf{K}}_{Y}- \mathbf{B}_{1}^{*}- \mathbf{B}_{2}^{*}\right)^{-1}\left(\mathbf{K} + {\mathbf{K}}_{Y}- \mathbf{B}_{1}^{*}\right).
$
\end{enumerate}

\end{lemma}

\begin{proof}
The proof of Lemma \ref{lemma_enhancement} is left in Appendix \ref{app_enhancement}.
\end{proof}
\smallskip

Moreover, to decouple our extremal inequality \eqref{eq:exinq}, we need the vector generalization of Costa's entropy power inequality \cite{Liu10}, and the generalized extremal inequality related the degraded compound MIMO Gaussian broadcast channel in \cite{WLSSV09,CL14}, as two auxiliary lemmas.

\begin{lemma}[ {\cite[Corollary 2]{Liu10}}] \label{lemma:Cos}
Let $\mathbf{Z}_{1}$, $\mathbf{Z}_{2}$ and $\mathbf{Z}_{3}$ be Gaussian random vectors with positive definite covariance matrices $\mathbf{N}_{1}$, $\mathbf{N}_{2}$ and $\mathbf{N}_{3}$, respectively. Furthermore, $\mathbf{N}_{1}$, $\mathbf{N}_{2}$ satisfy $\mathbf{N}_{1} \preceq \mathbf{N}_{2}$. If there exists a positive semi-definite covariance $\mathbf{B}^{*}$ such that
\begin{equation}
\left(  \mathbf{B}^{*} + \mathbf{N}_{1}  \right)^{-1}+  \lambda  \left(  \mathbf{B}^{*} + \mathbf{N}_{2}  \right)^{-1} = (\lambda+1) \left(  \mathbf{B}^{*} + \mathbf{N}_{3}  \right)^{-1},
\end{equation}
where $\lambda \geq 0$, then
\begin{align}
&h(\mathbf{X} + \mathbf{Z}_{1}|U) + \mu h(\mathbf{X} + \mathbf{Z}_{2}|U) - (\mu+1) h(\mathbf{X} + \mathbf{Z}_{3}|U) \nonumber \\
& \leq  \;\frac{1}{2} \log | \mathbf{B}^{*} + \mathbf{N}_{1}| + \frac{\mu}{2} \log | \mathbf{B}^{*} + \mathbf{N}_{2}| - \frac{\mu+1}{2}   \log | \mathbf{B}^{*} + \mathbf{N}_{3}|,
\end{align}
for any $(\mathbf{X},U)$ independent of $(\mathbf{Z}_{1},\mathbf{Z}_{2},\mathbf{Z}_{3})$.
\end{lemma}
\smallskip

\begin{lemma}[ {\cite[Corollary 4]{WLSSV09}}] \label{lemma:ex}
Let $\{\mathbf{Z}_{i}\}_{i=1}^{L_1}$ and $\{\mathbf{Z}_{j}\}_{j=1}^{L_2}$ be real Gaussian random vectors with positive definite covariance matrices $\{\mathbf{N}_{i}\}_{i=1}^{L_1}$ and $\{\mathbf{N}_{j}\}_{j=1}^{L_2}$, respectively. We assume that there exits a covariance matrix $\mathbf{N}^{*}$ such that
\begin{equation}
 \{\mathbf{N}_{i}\}_{i=1}^{L_1} \preceq \mathbf{N}^{*} \preceq  \{ \mathbf{N}_{j}\}_{j=1}^{L_2}.
\end{equation}
Furthermore, let $\mathbf{B}^{*}$ be a positive semi-definite covariance matrix such that
\begin{equation}
\sum_{i=1}^{L_{1}}\lambda_i \left(  \mathbf{B}^{*} + \mathbf{N}_{i}  \right)^{-1} = \sum_{j=1}^{L_{2}}\lambda_j \left(  \mathbf{B}^{*} + \mathbf{N}_{j}  \right)^{-1} + \mathbf{\Psi},
\end{equation}
where $\lambda_{i}, \lambda_{j} \geq 0$, $1 \leq i \leq L_1, 1 \leq j \leq L_2$, $\mathbf{\Psi} \succeq \mathbf{0}$ and
\begin{equation}
\left( \mathbf{K} - \mathbf{B}^{*}\right)\mathbf{\Psi} = \mathbf{\Psi}\left( \mathbf{K} - \mathbf{B}^{*}\right) = \mathbf{0}.
\end{equation}
The for any distribution $(\mathbf{X}, U)$ independent of $\{\mathbf{Z}_{i}\}_{i=1}^{L_1}$ and $\{\mathbf{Z}_{j}\}_{j=1}^{L_2}$ , such that $\cov \left(   \mathbf{X} | U    \right) \preceq \mathbf{K} $, we have
\begin{align}
&\sum_{i=1}^{L_{1}}\lambda_i h(\mathbf{X} + \mathbf{Z}_{i}|U)  - \sum_{j=1}^{L_2}\lambda_j h(\mathbf{X} + \mathbf{Z}_{j}|U)  \nonumber \\
& \leq  \;\sum_{i=1}^{L_1}\frac{\lambda_i}{2} \log | \mathbf{B}^{*} + \mathbf{N}_{i}|  - \sum_{j=1}^{L_2}\frac{\lambda_j}{2}   \log | \mathbf{B}^{*} + \mathbf{N}_{j}|.
\end{align}
\end{lemma}

\subsection{Proof of Theorem \ref{ext_thm}}
We proceed to proof the extremal inequality \eqref{eq:exinq}. At the beginning, we rewrite the \emph{l.h.s.} of \eqref{eq:exinq} by involving the enhancement source $\mathbf{\tilde{Y}}$.
\begin{subequations}
\begin{align}
&(\mu_1+\mu_2)h(\mathbf{Y}|U)- \mu_{1} h(\mathbf{Z}|U) - \mu_2 h(\mathbf{X}|U)\nonumber \\
&\;+  \mu_{1} h(\mathbf{Z}|V)    +       (\mu_3 - \mu_1) h(\mathbf{Y}|V) - \mu_3 h(\mathbf{X}|V ) \nonumber \\
&=(\mu_1+\mu_2)h(\tilde{\mathbf{Y}}|U)- \mu_{1} h(\mathbf{Z}|U) - \mu_2 h(\mathbf{X}|U)\label{eq:sub1} \\
& \quad\; +\mu_{1} h(\mathbf{Z}|V) + (\mu_{2} + \mu_{3}) h(\mathbf{Y}|V) - (\mu_1+\mu_2) h(\tilde{\mathbf{Y}}|V) - \mu_3 h(\mathbf{X}|V ) \label{eq:sub2} \\
& \quad\; +(\mu_1+\mu_2)\left(h(\mathbf{Y}|U)-h(\tilde{\mathbf{Y}}|U) \right)- (\mu_1+\mu_2)\left(h(\mathbf{Y}|V)-h(\tilde{\mathbf{Y}}|V) \right).\label{eq:sub3}
\end{align}
\end{subequations}

\begin{enumerate}
\item \emph{The lower bound of \eqref{eq:sub1}}: Notice that from the definition of $\mathbf{Y}$ in \eqref{eq:def_Y}, and the KKT conditions of \eqref{eq:KKT1}, we get
\begin{equation}
\frac{\mu_1}{2} \left(\mathbf{K} + \mathbf{K}_{Z}- \mathbf{B}_{1}^{*}-\mathbf{B}_{2}^{*} \right)^{-1}+ \frac{\mu_{2}}{2}  \left(  \mathbf{K}  - \mathbf{B}_{1}^{*}- \mathbf{B}_{2}^{*} \right)^{-1}=\frac{\mu_1+\mu_2}{2} \left(  \mathbf{K} + \tilde{\mathbf{K}}_{Y} - \mathbf{B}_{1}^{*}-\mathbf{B}_{2}^{*} \right)^{-1}.
\end{equation}
Furthermore, from the first and second statement of Lemma \ref{lemma_enhancement},
 \begin{equation}
 \mathbf{0}    \prec     \tilde{\mathbf{K}}_{Y}  \preceq {\mathbf{K}}_{Z}.
 \end{equation}
Using which in conjecture with Lemma \ref{lemma:Cos}, we have the following lower bound of \eqref{eq:sub1}:
\begin{align}
&(\mu_1+\mu_2)h(\tilde{\mathbf{Y}}|U)- \mu_{1} h(\mathbf{Z}|U) - \mu_2 h(\mathbf{X}|U)\nonumber \\
& \geq \frac{\mu_1+\mu_2}{2} \log \left|  \mathbf{K} + \tilde{\mathbf{K}}_{Y} - \mathbf{B}_{1}^{*}-\mathbf{B}_{2}^{*} \right|- \frac{\mu_1}{2} \log \left|\mathbf{K} + \mathbf{K}_{Z}- \mathbf{B}_{1}^{*}-\mathbf{B}_{2}^{*} \right| \nonumber \\
&\quad- \frac{\mu_{2}}{2} \log \left|  \mathbf{K}  - \mathbf{B}_{1}^{*}- \mathbf{B}_{2}^{*} \right|.
\end{align}

\item  \emph{The lower bound of \eqref{eq:sub2}}:  Notice that from the third statement of Lemma \ref{lemma_enhancement}, and the KKT conditions of \eqref{eq:KKT2}, we get
\begin{align}
&\frac{\mu_3}{2} \left( \mathbf{K}  - \mathbf{B}_{1}^{*} \right)^{-1} + \frac{\mu_1+\mu_2}{2}  \left(  \mathbf{K} + \tilde{\mathbf{K}}_{Y} - \mathbf{B}_{1}^{*} \right)^{-1}  \nonumber \\
&= \frac{\mu_2+\mu_3}{2}  \left(  \mathbf{K} + \mathbf{K}_{Y} - \mathbf{B}_{1}^{*} \right)^{-1}+ \frac{\mu_1}{2}  \left(  \mathbf{K} + \mathbf{K}_{Z} - \mathbf{B}_{1}^{*} \right)^{-1} + \mathbf{M}_{1},
\end{align}
Furthermore, from the first and second statement of Lemma \ref{lemma_enhancement},
 \begin{equation}
 \mathbf{0}    \prec     \tilde{\mathbf{K}}_{Y}  \preceq \left\{ {\mathbf{K}}_{Y}, {\mathbf{K}}_{Z} \right\}.
 \end{equation}
Notice the KKT conditions of \eqref{eq:KKT4} as well. We can use Lemma \ref{lemma:ex} to obtain the following lower bound on \eqref{eq:sub2}:
\begin{align}
& \mu_{1} h(\mathbf{Z}|V) + (\mu_{2} + \mu_{3}) h(\mathbf{Y}|V) - (\mu_1+\mu_2) h(\tilde{\mathbf{Y}}|V) - \mu_3 h(\mathbf{X}|V )\nonumber \\
& \geq  \frac{\mu_1}{2} \log \left|  \mathbf{K} + \mathbf{K}_{Z} - \mathbf{B}_{1}^{*} \right| +  \frac{\mu_{2}+\mu_3}{2} \log \left|  \mathbf{K}   + \mathbf{K}_{Y} - \mathbf{B}_{1}^{*} \right|\nonumber \\
&\quad- \frac{\mu_1+\mu_2}{2} \log \left|  \mathbf{K} + \tilde{\mathbf{K}}_{Y} - \mathbf{B}_{1}^{*} \right| - \frac{\mu_3}{2} \log \left|  \mathbf{K}  - \mathbf{B}_{1}^{*} \right|.
\end{align}

\item \emph{The lower bound of \eqref{eq:sub3}}:  From the first statement of Lemma \ref{lemma_enhancement}, we can decompose $\mathbf{Y}$ into two parts: $\mathbf{Y}= \tilde{\mathbf{Y}} + \hat{\mathbf{Y}}$, where $\hat{\mathbf{Y}}$ is a Gaussian random vector with covariance $\mathbf{K}_{Y}- \tilde{\mathbf{K}}_{Y}$, and it is assumed to be independent of $(\tilde{\mathbf{Y}},U)$. The difference between $h(\mathbf{Y}|U)$ and $h(\tilde{\mathbf{Y}} | U)$ may be written as
\begin{align}
& h(\mathbf{Y}|U)-h(\tilde{\mathbf{Y}} | U)\nonumber \\
&= h(\tilde{\mathbf{Y}}+ \hat{\mathbf{Y}} | U) - h(\tilde{\mathbf{Y}} | U) \\
&\overset{(a)}{=}  h(\hat{\mathbf{Y}})- h(\tilde{\mathbf{Y}} | \tilde{\mathbf{Y}}+ \hat{\mathbf{Y}}, U) \\
&\overset{(b)}{\geq}  h(\hat{\mathbf{Y}})- h(\tilde{\mathbf{Y}} | \tilde{\mathbf{Y}}+ \hat{\mathbf{Y}}, V) \\
&=  h({\mathbf{Y}}| V )-h(\tilde{\mathbf{Y}}| V). \label{tt}
\end{align}
where
\begin{itemize}
\item[(a)] follows from
$
  h(\tilde{\mathbf{Y}}, \tilde{\mathbf{Y}}+ \hat{\mathbf{Y}} | U) = h(\tilde{\mathbf{Y}}| U) + h(\hat{\mathbf{Y}})
  =h(\tilde{\mathbf{Y}}+ \hat{\mathbf{Y}} | U) + h(\tilde{\mathbf{Y}}| \tilde{\mathbf{Y}}+ \hat{\mathbf{Y}}, U ).
$
\item[(b)] follows from Markov chain $V \rightarrow U \rightarrow \mathbf{X} \rightarrow \tilde{\mathbf{Y}}$, and the fact that conditioning reduces entropy.
\end{itemize}

Thus, \eqref{tt} implies \eqref{eq:sub3} is lower bounded by 0.

\end{enumerate}

\smallskip

Now we are ready to show the extremal inequality \eqref{eq:exinq} via combining all the lower bounds together.
\begin{align}
&(\mu_1+\mu_2)h(\mathbf{Y}|U)- \mu_{1} h(\mathbf{Z}|U) - \mu_2 h(\mathbf{X}|U)\nonumber \\
&\;+  \mu_{1} h(\mathbf{Z}|V)    +       (\mu_3 - \mu_1) h(\mathbf{Y}|V) - \mu_3 h(\mathbf{X}|V ) \nonumber \\
& \geq \; \frac{\mu_1+\mu_2}{2} \log \left|  \mathbf{K} + \tilde{\mathbf{K}}_{Y} - \mathbf{B}_{1}^{*}-\mathbf{B}_{2}^{*} \right|- \frac{\mu_1}{2} \log \left|\mathbf{K} + \mathbf{K}_{Z}- \mathbf{B}_{1}^{*}-\mathbf{B}_{2}^{*} \right| \nonumber \\
&\quad- \frac{\mu_{2}}{2} \log \left|  \mathbf{K}  - \mathbf{B}_{1}^{*}- \mathbf{B}_{2}^{*} \right| + \frac{\mu_1}{2} \log \left|  \mathbf{K} + \mathbf{K}_{Z} - \mathbf{B}_{1}^{*} \right| +  \frac{\mu_{2}+\mu_3}{2} \log \left|  \mathbf{K}   + \mathbf{K}_{Y} - \mathbf{B}_{1}^{*} \right|\nonumber \\
&\quad- \frac{\mu_1+\mu_2}{2} \log \left|  \mathbf{K} + \tilde{\mathbf{K}}_{Y} - \mathbf{B}_{1}^{*} \right| - \frac{\mu_3}{2} \log \left|  \mathbf{K}  - \mathbf{B}_{1}^{*} \right| \\
& = \frac{\mu_1+\mu_2}{2}  \log \frac{\left|  \mathbf{K} + \tilde{\mathbf{K}}_{Y} - \mathbf{B}_{1}^{*}-\mathbf{B}_{2}^{*} \right|}{\left|  \mathbf{K} + \tilde{\mathbf{K}}_{Y} - \mathbf{B}_{1}^{*} \right|} - \frac{\mu_1}{2} \log \left|\mathbf{K} + \mathbf{K}_{Z}- \mathbf{B}_{1}^{*}-\mathbf{B}_{2}^{*} \right| \nonumber \\
&\quad- \frac{\mu_{2}}{2} \log \left|  \mathbf{K}  - \mathbf{B}_{1}^{*}- \mathbf{B}_{2}^{*} \right| + \frac{\mu_1}{2} \log \left|  \mathbf{K} + \mathbf{K}_{Z} - \mathbf{B}_{1}^{*} \right| +  \frac{\mu_{2}+\mu_3}{2} \log \left|  \mathbf{K}   + \mathbf{K}_{Y} - \mathbf{B}_{1}^{*} \right|\nonumber \\
& \quad - \frac{\mu_3}{2} \log \left|  \mathbf{K}  - \mathbf{B}_{1}^{*} \right| \\
& \overset{(a)}= \frac{\mu_1+\mu_2}{2}  \log \frac{\left|  \mathbf{K} + {\mathbf{K}}_{Y} - \mathbf{B}_{1}^{*}-\mathbf{B}_{2}^{*} \right|}{\left|  \mathbf{K} + {\mathbf{K}}_{Y} - \mathbf{B}_{1}^{*} \right|} - \frac{\mu_1}{2} \log \left|\mathbf{K} + \mathbf{K}_{Z}- \mathbf{B}_{1}^{*}-\mathbf{B}_{2}^{*} \right| \nonumber \\
&\quad- \frac{\mu_{2}}{2} \log \left|  \mathbf{K}  - \mathbf{B}_{1}^{*}- \mathbf{B}_{2}^{*} \right| + \frac{\mu_1}{2} \log \left|  \mathbf{K} + \mathbf{K}_{Z} - \mathbf{B}_{1}^{*} \right| +  \frac{\mu_{2}+\mu_3}{2} \log \left|  \mathbf{K}   + \mathbf{K}_{Y} - \mathbf{B}_{1}^{*} \right|\nonumber \\
& \quad - \frac{\mu_3}{2} \log \left|  \mathbf{K}  - \mathbf{B}_{1}^{*} \right| \\
&=\; \frac{\mu_1+\mu_2}{2} \log \left|  \mathbf{K} + \mathbf{K}_{Y} - \mathbf{B}_{1}^{*}-\mathbf{B}_{2}^{*} \right|- \frac{\mu_1}{2} \log \left|\mathbf{K} + \mathbf{K}_{Z}- \mathbf{B}_{1}^{*}-\mathbf{B}_{2}^{*} \right| \nonumber \\
&\quad- \frac{\mu_{2}}{2} \log \left|  \mathbf{K}  - \mathbf{B}_{1}^{*}- \mathbf{B}_{2}^{*} \right| + \frac{\mu_1}{2} \log \left|  \mathbf{K} + \mathbf{K}_{Z} - \mathbf{B}_{1}^{*} \right| \nonumber \\
&\quad+ \frac{\mu_3-\mu_1}{2} \log \left|  \mathbf{K} + \mathbf{K}_{Y} - \mathbf{B}_{1}^{*} \right| - \frac{\mu_3}{2} \log \left|  \mathbf{K}  - \mathbf{B}_{1}^{*} \right|,
\end{align}
where (a) follows from the forth statement of Lemma \ref{lemma_enhancement}.

\section{Conclusion}

In this paper, we consider the vector Gaussian secret key generation problem with limited rate constrained public and private communications. A single-letter characterization is obtained. The proof is based on suitable applications of the enhancement argument, and the proof of a new extremal inequality, which should be decoupled into two different forms of extremal inequalities. One future direction of our work is the Vector Gaussian secure source problem in \cite{EU13}. Since the covariance matrices constraints make that the enhancement argument fails to establish the degraded order of sources, the single-letter characterization of that problem is still open. A new extremal inequality construction method should be investigated in a future.

\appendices
\section{The Achievability of Gaussian Codebooks in Theorem \ref{mainth2}} \label{app:0}
The achievability part is to show the rate region of \eqref{TH1}-\eqref{TH11} can be parameterized by two positive semi-definite matrices $\mathbf{B}_{1}, \mathbf{B}_{2} \succeq \mathbf{0}$. The ingredient here is that the Markov chain $V \rightarrow U \rightarrow \mathbf{X} \rightarrow (\mathbf{Y}, \mathbf{Z})$ can be transferred into certain partial order relations on Gaussian covariance matrices. Moreover, we can find particular a point parameterized by $(\mathbf{B}_{1}, \mathbf{B}_{2})$ in the rate region of \eqref{TH1}-\eqref{TH11}.

By restricting to the Gaussian distributed codewords, the jointly random vectors $(U, V, \mathbf{X})$ can be expressed in the following form
\begin{align}
V = \mathbf{X} + \mathbf{N}_{V}, \\
U = \mathbf{X} + \mathbf{N}_{U},
\end{align}
where $\mathbf{N}_{V}, \mathbf{N}_{U}$ are zero mean Gaussian random vectors with covariance matrices $\mathbf{\Sigma}_{V}, \mathbf{\Sigma}_{U}$. Notice the $\mathbf{N}_{V}, \mathbf{N}_{U}$ is independent of $\mathbf{X}, \mathbf{Y}, \mathbf{Z}$, but can be correlated to each other. For the sake of simplifying notations in the rest derivations, we write
\begin{align}
\mathbf{K}_{X|V} = \left(\mathbf{K}^{-1} + \mathbf{N}_{V}^{-1}\right)^{-1},\\
\mathbf{K}_{X|U} = \left( \mathbf{K}^{-1} + \mathbf{N}_{U}^{-1}\right)^{-1}.
\end{align}
Equations \eqref{TH1}-\eqref{TH11} in Theorem \ref{mainth} can be shown as
\begin{align}
R_K-R_{2} & \leq I(U;\mathbf{Y}|V)-I(U;\mathbf{Z}|V) \\
&\overset{(a)}= h(\mathbf{Y}|V)-h(\mathbf{Z}|V)-h(\mathbf{Y}|U)+h(\mathbf{Z}|U) \\
& = \frac{1}{2} \log \frac{|\mathbf{K}_{X|V} + \mathbf{K}_{Y}|}{|\mathbf{K}_{X|V} + \mathbf{K}_{Z}|} - \frac{1}{2} \log \frac{|\mathbf{K}_{X|U} + \mathbf{K}_{Y}|}{|\mathbf{K}_{X|U} + \mathbf{K}_{Z}|}\\
& \triangleq I_{1} \left(\mathbf{K_{X|V}}, \mathbf{K_{X|U}}\right),
\end{align}
where (a) follows from Makov Chain $V \rightarrow U \rightarrow \mathbf{X} \rightarrow (\mathbf{Y}, \mathbf{Z})$.
\begin{align}
R_1+R_2 &\ge I(U;\mathbf{X}|\mathbf{Y})\\
& \overset{(a)}= h(\mathbf{X}| U) - h (\mathbf{Y} | U) \\
& =\frac{1}{2} \log \frac{|\mathbf{K}_{X|U}|}{|\mathbf{K}_{X|U} + \mathbf{K}_{Y}|}\\
& \triangleq I_{2} \left(  \mathbf{K}_{X|U} \right),
\end{align}
where (a) follows from Makov Chain $U \rightarrow \mathbf{X} \rightarrow \mathbf{Y}$.
\begin{align}
R_1 &\ge I(V;\mathbf{X}|\mathbf{Y})\\
& \overset{(a)}= h(\mathbf{X}| V) - h (\mathbf{Y} | V) \\
& =\frac{1}{2} \log \frac{|\mathbf{K}_{X|V}|}{|\mathbf{K}_{X|V} + \mathbf{K}_{Y}|}\\
& \triangleq I_{3} \left(  \mathbf{K}_{X|V} \right),
\end{align}
where (a) follows from Makov Chain $V \rightarrow \mathbf{X} \rightarrow \mathbf{Y}$.
Moreover, because of the Markov chain $V \rightarrow U \rightarrow \mathbf{X}$, we have the following partial order on the positive definite matrices $\mathbf{K}_{X|V}, \mathbf{K}_{X|U}$,
\begin{align}
\mathbf{K} \succeq \mathbf{K}_{X|V} \succeq \mathbf{K}_{X|U} \succ \mathbf{0}.
\end{align}

Now consider the optimization problem
\begin{align}\label{maxx}
\max & \quad I_{1} \left(\mathbf{K_{X|V}}, \mathbf{K_{X|U}}\right) \nonumber \\
\text{subject to} & \quad I_{2}\left(  \mathbf{K}_{X|U} \right) \leq R_{1} + R_{2}, \nonumber \\
{} & \quad I_{3}\left(  \mathbf{K}_{X|V} \right) \leq R_1, \nonumber \\
{} & \quad \mathbf{K} \succeq \mathbf{K}_{X|V} \succeq \mathbf{K}_{X|U} \succ \mathbf{0}.
\end{align}
Let $(      \mathbf{K}^{*}_{X|V}, \mathbf{K}^{*}_{X|U}    )$ be an optimal solution of this problem. Since $I_{2}(\cdot), I_{3}(\cdot)$ should be less than some nonnegative values, it indicate that there exist a positive definite matrix $\mathbf{L} \succ \mathbf{0}$ such that $\mathbf{K}^{*}_{X|V}, \mathbf{K}^{*}_{X|U} \succeq \mathbf{L}$. From this view of point, the constraint $\mathbf{K}_{X|U} \succ \mathbf{0}$ can be removed safely, without affecting the optimal values of maximization problem \eqref{maxx}.

Furthermore, $(      \mathbf{K}^{*}_{X|V}, \mathbf{K}^{*}_{X|U}    )$ can be denoted by
\begin{align}
\mathbf{K}^{*}_{X|V} &= \mathbf{K} - \mathbf{B}_{1}, \label{qq1}\\
\mathbf{K}^{*}_{X|U} &= \mathbf{K} - \mathbf{B}_{1}-\mathbf{B}_{2}, \label{qq2}
\end{align}
and substitute \eqref{qq1} and \eqref{qq2} into $I_{1}(\cdot), I_{2}(\cdot)$ and $I_{3}(\cdot)$. Inequalities \eqref{eq:alin1}-\eqref{eq:align3} in Theorem \ref{mainth2} are obtained equivalently for some positive semi-definite matrices $\mathbf{B}_{1}$, $\mathbf{B}_{2} \succeq \mathbf{0}$.

\section{The Achievability of Theorem \ref{mainth}} \label{app:1}
Consider a given distribution
\begin{equation}
  p(v,u,x,y,z)=p(v,u)p(x|u)p(y,z|x).
\end{equation}

\emph{Case 1}: If $R_1\geq I(U;X|Y)$

We use a separate scheme as follows. Alice use the public channel to achieve a secrecy rate $I(U;Y|V)-I(U;Z|V)$ by the scheme in secret key generation, and use the secure channel to distribute another secret key with rate $R_2$. Combine two secret keys together,  the secrecy rate in (\ref{TH1}) is achievable.

\emph{Case 2}: If $R_1\leq I(U;X|Y)$.

In this case, we assume that $I(V;Y)\le I(V;Z)$. The reason is that if $I(V;Y)>I(V;Z)$, we define a pair of new random variable $(U',V')$ such that $U'=(U,V)$ and $V'=\emptyset$. Then we have $I(V';Y)\le I(V';Z)$. From the three inequalities in the main theorem, we have
\begin{align}
R_K&\le I(U;Y|V)-I(U;Z|V)+R_2\\
&\le I(U';Y|V')-I(U';Z|V')+R_2\\
  R_1+R_2 &\ge I(U;X|Y)=I(U';X|Y)  \\
  R_1&\ge I(V;X|Y)\ge I(V';X|Y)
\end{align}
The above derivation shows that for every pair $(U,V)$ with Markov chain $ V' \rightarrow U' \rightarrow X $ and  $I(V;Y)>I(V;Z)$, we can find another pair $(U',V')$ with Markov chain $ V \rightarrow U \rightarrow X $ and  $I(V';Y)\le I(V';Z)$ such that the rate region  with $(U',V')$ is achievable implies that the rate region with $(U,T)$ is also achievable. Therefore, we only need to consider the case where $I(V;Y)\le I(V;Z)$. 

We define the following notations
\begin{align}
    R_{11} &= I(V;X|Y) \\
    R_{12} &= R_1-R_{11} \\
    R_{21} &= I(U;X|Y,V)-R_{12} \\
    R_{22} &= R_2-R_{21} \\
    R_V &= I(V;X)\\
    R_U &= I(U;X|V)\\
    R_{K_1} &= I(U;Y|V)-I(U;Z|V)
\end{align}

\emph{Codebook generation}:
\begin{enumerate}
  \item Randomly and independently generation $2^{nR_V}$ sequences $v^n(s_1)$, $s_1\in[1:2^{nR_{v}}]$ according to distribution $p(v)$ and randomly and independently partition them into $2^{nR_{11}}$ bins with bin indices $\mathcal{B}_1(m_{11})$, $m_{11}\in[1:2^{nR_{11}}]$.
  \item For each codeword $v^n(s_1)$, randomly and independently generation $2^{nR_U}$ sequences $u^n(s_1,s_2)$, $s_2\in[1:2^{nR_{U}}]$ according to $p(u|v)$, and randomly and independently partition them into $2^{n(R_{12}+R_{21})}$ bins with bin indices $\mathcal{B}_2(s_1,m_{12},m_{21})$, $m_{12}\in[1:2^{nR_{12}}]$ and $m_{21}\in[1:2^{nR_{21}}]$. Randomly and independently partition the sequences in each nonempty bin $\mathcal{B}_2(s_1,m_{12},m_{21})$ into $2^{nR_{K_1}}$ bins with bin indices  $\mathcal{B}_2(s_1,m_{12},m_{21},k_1)$, $k_1\in[1:2^{nR_{K_1}}]$.
\end{enumerate}

\emph{Encoding}:
\begin{enumerate}
  \item Given a source sequence $X^n$, find the index $s_1$ such that $(v^n(s_1),X^n)$ is jointly typical with respect to the joint distribution $p(v,x)$, i.e.,  $(v^n(s_1),X^n)\in \mathcal{T}^n_{[VX]_\delta}$. If there is no such index or there are more than one such indices, then randomly select $s_1$ from $[1:2^{nR_{V}}]$, where the probability of such event is arbitrarily small if $R_V > I(V;X)$. 
  Let $\mathcal{B}_1(m_{11})$ be the bin index of $v^n(s_1)$. 
  \item Then find the index $s_2$ such that $(v^n(s_1),u^n(s_1,s_2),X^n)\in \mathcal{T}^n_{[VUX]_\delta}$.  If there is no such index or there are more than one such indices, then random select $s_2$ from $[1:2^{nR_{U}}]$, where the probability of such event is arbitrarily small if $R_U > I(U;X|V)$. Let $\mathcal{B}_2(s_1,m_{12},m_{21},k_1)$ be the bin index of $u^n(s_1,s_2)$.
  \item Randomly choose $k_2\in[1;2^{nR_{22}}]$.
  \item Alice sends $M_1=(m_{11},m_{12})$ and $M_2=(m_{21},k_2)$ to Bob through the public and secure channel, respectively.
  \item Alice chooses $K=(k_1,k_2,m_{21})$ as the secret key.
\end{enumerate}

\emph{Decoding}:  Upon receiving $(M_1,M_2)$, Bob decodes the secret key as following:
\begin{enumerate}
  \item Find the unique sequence $v^n(\hat{s}_1)\in \mathcal{B}_1(m_{11})$ such that $(v^n(\hat{s}_1),y^n)\in \mathcal{T}^n_{[VY]_\delta}$. And the probability that there is no such sequence or there are more than one such sequences is arbitrarily small if
      \begin{equation}
        R_V-R_{11} < I(V;Y)
      \end{equation}
  \item Find the index $\hat{s}_2\in \mathcal{B}_2(\hat{s}_1,m_{12},m_{21})$  such that $(v^n(\hat{s}_1),u^n(\hat{s}_1,\hat{s}_2),y^n)\in \mathcal{T}^n_{[VUY]_\delta}$. Let $\mathcal{B}_2(\hat{s}_1,m_{12},m_{21},\hat{k}_1)$ be the bin index of $u^n(\hat{s}_1,\hat{s}_2)$. The probability that there is no such sequence or there are more than one such sequences is arbitrarily small if
      \begin{equation}
        R_U-R_{12}-R_{21} < I(U;Y|V)
      \end{equation}
  \item Let $\hat{K}=(\hat{k}_1,k_2,m_{21})$
\end{enumerate}

\emph{Key leakage rate}: We consider the key leakage rate averaged over the random coding scheme as described above. We begin with
\begin{align}
  &H(K|Z^n,M_1) \nonumber\\
   &= H(K,U^n,V^n|Z^n,M_1)-H(U^n,V^n|Z^n,K,M_1)  \\
   &= H(U^n,V^n|Z^n,M_1)+H(K|U^n,V^n,Z^n,M_1)-H(U^n,V^n|Z^n,K,M_1)\\
   &\overset{(a)}= H(V^n|Z^n,M_1)+H(U^n|Z^n,M_1,V^n)+H(K_2)-H(V^n|K,Z^n,M_1)-H(U^n|K,Z^n,M_1,V^n)\label{k2}\\
   &= I(V^n;K|Z^n,M_1)+H(U^n|Z^n,V^n)+H(K_2)-I(U^n;M_1|Z^n,V^n)-H(U^n|K,Z^n,M_1,V^n)\\
   &\ge H(U^n|Z^n,V^n)+H(K_2)+H(M_1|Z^n,V^n,U^n)-H(M_1|Z^n)-H(U^n|K,Z^n,M_1,V^n)\\
   &\overset{(b)}\geq H(U^n|Z^n,V^n)+H(K_2)-H(M_1)-H(U^n|K,Z^n,M_1,V^n)\label{w1}\\
   &\overset{(c)}\ge H(U^n|Z^n,V^n)+nR_{22}-nR_1-H(U^n|K_1,M_{21},Z^n,M_1,V^n)\label{final}
\end{align}
where
\begin{enumerate}
\item[(a)] follows because $K=(K_1,K_2, M_{21})$, in which $K_2$ is independent of $(U^n,V^n,Z^n,M_1)$ and $K_1$ and $M_{21}$ are functions of $U^n$;
\item[(b)] follows because $M_1=(M_{11},M_{12})$ is a function of $(V^n,U^n)$;
\item[(c)] follows because $K_2$ is uniformly distributed, $H(M_1)$ is upper bounded by $nR_1$, and $K=(K_1,M_{21})$.
\end{enumerate}


For the first term in (\ref{final}), we have
\begin{align}
 & \frac{1}{n}H(U^n|Z^n,V^n) \nonumber \\
  &=\frac{1}{n}[H(U^n|V^n)-I(Z^n;U^n|V^n)]\\
  &\leq R_U-I(Z;U|V)
\end{align}


For the fourth term, similar to \cite[Lemma 22.3]{el2011network}, we can bound it in the following lemma.
\begin{lemma}
If $R_{K_1}< H(U|Z,V)-H(U|Y,V)-2\delta(\epsilon)$
\begin{align}
  \limsup_{n\rightarrow \infty} & \frac{1}{n} H(U^n|K_1,M_{21},Z^n,M_1,V^n) \leq R_U-R_{K_1}-R_{21}-R_1-I(Z;U|V)+\delta(\epsilon)
\end{align}
\end{lemma}

From the above derivation,  we have
\begin{align}
  &\frac{1}{n}I(K;Z^n,M_1)\nonumber \\
  &=\frac{1}{n} [H(K)-H(K|Z^n,M_1)] \\
  &\leq \frac{1}{n}H(K)-[R_{K_1}+R_{21}+R_{22}-\delta(\epsilon)]\\
  &\le \delta(\epsilon)
\end{align}
which concludes the proof of the achievability of Theorem \ref{mainth}.

\section{Proof of Lemma \ref{lemma_KKT}}\label{KKT}

Since the program of $\mathfrak{R}^{*}(\mu_{1}, \mu_{2}, \mu_{3})$ in \eqref{eqn:opt} is not necessarily convex, the KKT conditions are necessary but not sufficient. We first give the Lagrangian
of \eqref{eqn:opt} as follows
\begin{align}
\mathcal{L} = &\frac{\mu_1+\mu_2}{2} \log \left|  \mathbf{K} + \mathbf{K}_{Y} - \mathbf{B}_{1}-\mathbf{B}_{2} \right|- \frac{\mu_1}{2} \log \left|\mathbf{K} + \mathbf{K}_{Z}- \mathbf{B}_{1}-\mathbf{B}_{2} \right| \nonumber \\
&- \frac{\mu_{2}}{2} \log \left|  \mathbf{K}  - \mathbf{B}_{1}- \mathbf{B}_{2} \right| + \frac{\mu_1}{2} \log \left|  \mathbf{K} + \mathbf{K}_{Z} - \mathbf{B}_{1} \right| \nonumber \\
&+ \frac{\mu_3-\mu_1}{2} \log \left|  \mathbf{K} + \mathbf{K}_{Y} - \mathbf{B}_{1} \right| - \frac{\mu_3}{2} \log \left|  \mathbf{K}  - \mathbf{B}_{1} \right| \nonumber \\
&+ \frac{\mu_2+\mu_3}{2} \log |\mathbf{K}| -\frac{\mu_2+\mu_3}{2} \log |\mathbf{K}+\mathbf{K}_{Y}|  \nonumber \\
& - \tr \left\{    \mathbf{B}_{1}\mathbf{M}_{1} + \mathbf{B}_{2}\mathbf{M}_{2}     \right\},
\end{align}
where $\mathbf{M}_{1} $ and $\mathbf{M}_{2} $ are Lagrange multipliers. Let $(\mathbf{B}_{1}^{*}, \mathbf{B}_{2}^{*})$ be one optimal solution of $\mathfrak{R}^{*}(\mu_{1}, \mu_{2}, \mu_{3})$. The necessary KKT conditions that it need to satisfy are given by
\begin{align}
\nabla_{\mathbf{B}_{1}}\mathcal{L} =& \frac{\mu_1+\mu_2}{2} \left(  \mathbf{K} + \mathbf{K}_{Y} - \mathbf{B}_{1}^{*}-\mathbf{B}_{2}^{*} \right)^{-1}- \frac{\mu_1}{2} \left(\mathbf{K} + \mathbf{K}_{Z}- \mathbf{B}_{1}^{*}-\mathbf{B}_{2}^{*} \right)^{-1}- \frac{\mu_{2}}{2}  \left(  \mathbf{K}  - \mathbf{B}_{1}^{*}- \mathbf{B}_{2}^{*} \right)^{-1} \nonumber \\
&+ \frac{\mu_1}{2}  \left(  \mathbf{K} + \mathbf{K}_{Z} - \mathbf{B}_{1}^{*} \right)^{-1}+ \frac{\mu_3-\mu_1}{2}  \left(  \mathbf{K} + \mathbf{K}_{Y} - \mathbf{B}_{1}^{*} \right)^{-1} - \frac{\mu_3}{2} \left( \mathbf{K}  - \mathbf{B}_{1}^{*} \right)^{-1} + \mathbf{M}_{1} = \mathbf{0}, \label{eq:KK1} \\
\nabla_{\mathbf{B}_{2}}\mathcal{L} =& \frac{\mu_1+\mu_2}{2} \left(  \mathbf{K} + \mathbf{K}_{Y} - \mathbf{B}_{1}^{*}-\mathbf{B}_{2}^{*} \right)^{-1}- \frac{\mu_1}{2} \left(\mathbf{K} + \mathbf{K}_{Z}- \mathbf{B}_{1}^{*}-\mathbf{B}_{2}^{*} \right)^{-1}- \frac{\mu_{2}}{2}  \left(  \mathbf{K}  - \mathbf{B}_{1}^{*}- \mathbf{B}_{2}^{*} \right)^{-1} + \mathbf{M}_{2} = \mathbf{0}, \label{eq:KK2}
\end{align}
\begin{align}
\tr \left\{\mathbf{B}_{1}^{*}\mathbf{M}_{1} \right\} =& \boldsymbol{0}, \label{KK3}\\
\tr \left\{\mathbf{B}_{2}^{*}\mathbf{M}_{2} \right\}=& \boldsymbol{0}, \label{KK4}\\
\mathbf{B}_{1}^{*}, \mathbf{B}_{2}^{*}, \mathbf{M}_{1}, \mathbf{M}_{2} \succeq & \boldsymbol{0}.
\end{align}
Notice that \eqref{eq:KK1} is equivalent to \eqref{eq:KKT1}. Substituting \eqref{eq:KK2} from \eqref{eq:KK1} yields \eqref{eq:KKT2}.Since $\tr \left\{ \mathbf{A} \mathbf{B}\right\} = \tr \left\{ \mathbf{B} \mathbf{A}\right\} \geq 0$, for $\mathbf{A}, \mathbf{B} \succeq \boldsymbol{0}$, \eqref{KK3}-\eqref{KK4} yields \eqref{eq:KKT4}-\eqref{eq:KKT5}.

\section{Proof of Lemma \ref{lemma_enhancement}} \label{app_enhancement}
\begin{enumerate}
\item

From the definition of $\tilde{\mathbf{K}}_{Y}$ in \eqref{eq:def_Y},
\begin{align}
&\frac{\mu_1+\mu_2}{2}\left(\mathbf{K}+\tilde{\mathbf{K}}_{Y} - \mathbf{B}_{1}^{*} - \mathbf{B}_{2}^{*}\right)^{-1}
 = \frac{\mu_1+\mu_2}{2}\left(\mathbf{K} + \mathbf{K}_{Y} - \mathbf{B}_{1}^{*} - \mathbf{B}_{2}^{*}\right)^{-1} + \mathbf{M}_{2} \\
& \succeq  \frac{\mu_1+\mu_2}{2} \left(\mathbf{K} + \mathbf{K}_{Y} - \mathbf{B}_{1}^{*} - \mathbf{B}_{2}^{*}\right)^{-1}.\label{eq:tmp1}
\end{align}
The last inequality is due to $\mathbf{M}_2 \succeq \mathbf{0}$. Thus we have
\begin{equation}
\mathbf{K} + \tilde{\mathbf{K}}_{Y} - \mathbf{B}_{1}^{*} - \mathbf{B}_{2}^{*}\preceq \mathbf{K}+ {\mathbf{K}}_{Y} - \mathbf{B}_{1}^{*}-\mathbf{B}_{2}^{*}.
\end{equation}
This yields
\begin{equation}
\tilde{\mathbf{K}}_{Y} \preceq {\mathbf{K}}_{Y}.
\end{equation}

Sine $\mathbf{K}_{Z} \succ \mathbf{0}$, by substituting \eqref{eq:def_Y} into KKT conditions of \eqref{eq:KKT1}, we get
\begin{align}
& \frac{\mu_1+\mu_2}{2} \left(  \mathbf{K} + \tilde{\mathbf{K}}_{Y} - \mathbf{B}_{1}^{*}-\mathbf{B}_{2}^{*} \right)^{-1} = \frac{\mu_1}{2} \left(\mathbf{K} + \mathbf{K}_{Z}- \mathbf{B}_{1}^{*}-\mathbf{B}_{2}^{*} \right)^{-1}+ \frac{\mu_{2}}{2}  \left(  \mathbf{K}  - \mathbf{B}_{1}^{*}- \mathbf{B}_{2}^{*} \right)^{-1} \label{eq:ttm}\\
& \prec \frac{\mu_1+\mu_2}{2} \left(  \mathbf{K}  - \mathbf{B}_{1}^{*}-\mathbf{B}_{2}^{*} \right)^{-1}.
\end{align}
Thus we have
\begin{equation}
\mathbf{K}  - \mathbf{B}_{1}^{*} - \mathbf{B}_{2}^{*}\prec \mathbf{K} + \tilde{\mathbf{K}}_{Y} - \mathbf{B}_{1}^{*}-\mathbf{B}_{2}^{*}.
\end{equation}
This yields
\begin{equation}
\tilde{\mathbf{K}}_{Y} \succ {\mathbf{0}}.
\end{equation}

\smallskip
\item

From \eqref{eq:ttm}, we have
\begin{align}
\frac{\mu_1+\mu_2}{2} \left(  \mathbf{K} + \tilde{\mathbf{K}}_{Y} - \mathbf{B}_{1}^{*}-\mathbf{B}_{2}^{*} \right)^{-1} \succeq \frac{\mu_1+\mu_2}{2} \left(\mathbf{K} + \mathbf{K}_{Z}- \mathbf{B}_{1}^{*}-\mathbf{B}_{2}^{*} \right)^{-1}.
\end{align}
 Thus we have
\begin{equation}
\mathbf{K} + \tilde{\mathbf{K}}_{Y} - \mathbf{B}_{1}^{*} - \mathbf{B}_{2}^{*}\preceq \mathbf{K}+ {\mathbf{K}}_{Z} - \mathbf{B}_{1}^{*}-\mathbf{B}_{2}^{*}.
\end{equation}
This yields
\begin{equation}
\tilde{\mathbf{K}}_{Y} \preceq {\mathbf{K}}_{Z}.
\end{equation}

\smallskip
\item

The proof of Statement 3) is shown as follows:
\begin{align}
& \mathbf{K} + \tilde{\mathbf{K}}_{Y} -\mathbf{B}_{1}^{*} \nonumber \\
\overset{(a)}= &  \left(\left(\mathbf{K} + {\mathbf{K}}_{Y} -\mathbf{B}_{1}^{*}-\mathbf{B}_{2}^{*}\right)^{-1}+\frac{1}{\mu_1+\mu_2}\mathbf{M}_{2} \right)^{-1} +\mathbf{B}_{2}^{*}\\
=&\left(\mathbf{I}+\frac{1}{\mu_1+\mu_2}\left(\mathbf{K} + {\mathbf{K}}_{Y} -\mathbf{B}_{1}^{*}-\mathbf{B}_{2}^{*}\right)\mathbf{M}_{2} \right)^{-1}\left(\mathbf{K} + {\mathbf{K}}_{Y} -\mathbf{B}_{1}^{*}-\mathbf{B}_{2}^{*}\right)+\mathbf{B}_{2}^{*} \\
\overset{(b)}=&\left(\mathbf{I}+\frac{1}{\mu_1+\mu_2}\left(\mathbf{K} + {\mathbf{K}}_{Y} -\mathbf{B}_{1}^{*}\right)\mathbf{M}_{2} \right)^{-1}\left(\mathbf{K} + {\mathbf{K}}_{Y} -\mathbf{B}_{1}^{*}-\mathbf{B}_{2}^{*}\right)+\mathbf{B}_{2}^{*} \\
=&\left(\left(\mathbf{K} + {\mathbf{K}}_{Y} -\mathbf{B}_{1}^{*}\right)^{-1}+\frac{1}{\mu_1+\mu_2}\mathbf{M}_{2} \right)^{-1}\left(\mathbf{K} + {\mathbf{K}}_{Y} -\mathbf{B}_{1}^{*}\right)^{-1}\left(\mathbf{K} + {\mathbf{K}}_{Y} -\mathbf{B}_{1}^{*}-\mathbf{B}_{2}^{*}\right) +\mathbf{B}_{2}^{*}\\
=&\left(\left(\mathbf{K} + {\mathbf{K}}_{Y} -\mathbf{B}_{1}^{*}\right)^{-1}+\frac{1}{\mu_1+\mu_2}\mathbf{M}_{2} \right)^{-1}\nonumber \\
& \;-\left(\left(\mathbf{K} + {\mathbf{K}}_{Y} -\mathbf{B}_{1}^{*}\right)^{-1}+\frac{1}{\mu_1+\mu_2}\mathbf{M}_{2} \right)^{-1}\left(\mathbf{K} + {\mathbf{K}}_{Y} -\mathbf{B}_{1}^{*}\right)^{-1}\mathbf{B}_{2}^{*}+\mathbf{B}_{2}^{*} \\
\overset{(c)}=&\left(\left(\mathbf{K} + {\mathbf{K}}_{Y} -\mathbf{B}_{1}^{*}\right)^{-1}+\frac{1}{\mu_1+\mu_2}\mathbf{M}_{2} \right)^{-1} \nonumber \\
& \;-\left(\left(\mathbf{K} + {\mathbf{K}}_{Y} -\mathbf{B}_{1}^{*}\right)^{-1}+\frac{1}{\mu_1+\mu_2}\mathbf{M}_{2} \right)^{-1}\left(\left(\mathbf{K} +{\mathbf{K}}_{Y} -\mathbf{B}_{1}^{*}\right)^{-1}+\frac{1}{\mu_1+\mu_2}\mathbf{M}_{2}\right)\mathbf{B}_{2}^{*} +\mathbf{B}_{2}^{*}\\
=&\left(\left(\mathbf{K} + {\mathbf{K}}_{Y} -\mathbf{B}_{1}^{*}\right)^{-1}+\frac{1}{\mu_1+\mu_2}\mathbf{M}_{2} \right)^{-1}-\mathbf{B}_{2}^{*}+\mathbf{B}_{2}^{*} \\
=&\left(\left(\mathbf{K} + {\mathbf{K}}_{Y} -\mathbf{B}_{1}^{*}\right)^{-1}+\frac{1}{\mu_1+\mu_2}\mathbf{M}_{2} \right)^{-1},
\end{align}
where
\begin{enumerate}
\item[(a)] comes from the definition in \eqref{eq:def_Y};
\item[(b)] is due to KKT condition \eqref{eq:KKT5}: $\mathbf{B}_{2}^{*} \mathbf{M}_{2} = \mathbf{M}_{2}\mathbf{B}_{2}^{*} =\mathbf{0}$;
\item[(c)] is also due to KKT condition \eqref{eq:KKT5}.
\end{enumerate}
\smallskip
\item
The proof of Statement 4) is shown as follows:
\begin{align}
&\left(\mathbf{K} + {\mathbf{K}}_{Y}- \mathbf{B}_{1}^{*}- \mathbf{B}_{2}^{*}\right)^{-1}\left(\mathbf{K} + {\mathbf{K}}_{Y}- \mathbf{B}_{1}^{*}\right) \nonumber \\
=&\mathbf{I} + \left(\mathbf{K} + {\mathbf{K}}_{Y}- \mathbf{B}_{1}^{*}- \mathbf{B}_{2}^{*}\right)^{-1}\mathbf{B}_{2}^{*}                      \\
\overset{(a)}=&\mathbf{I} + \left(\left(\mathbf{K} + {\mathbf{K}}_{Y}- \mathbf{B}_{1}^{*}- \mathbf{B}_{2}^{*}\right)^{-1}+\frac{1}{\mu+1}\mathbf{M}_{2}\right)\mathbf{B}_{2}^{*}   \\
\overset{(b)}=&\mathbf{I} + \left(\mathbf{K} + \tilde{\mathbf{K}}_{Y}- \mathbf{B}_{1}^{*}- \mathbf{B}_{2}^{*}\right)^{-1}\mathbf{B}_{2}^{*}   \\
=& \left(\mathbf{K} + \tilde{\mathbf{K}}_{Y}- \mathbf{B}_{1}^{*}- \mathbf{B}_{2}^{*}\right)^{-1}\left(\mathbf{K} + \tilde{\mathbf{K}}_{Y}- \mathbf{B}_{1}^{*}\right),
\end{align}
where
\begin{enumerate}
\item[(a)] is due to KKT condition \eqref{eq:KKT5}: $\mathbf{B}_{2}^{*} \mathbf{M}_{2} = \mathbf{M}_{2}\mathbf{B}_{2}^{*} =\mathbf{0}$;
\item[(b)] comes from the definition in \eqref{eq:def_Y}.
\end{enumerate}
\end{enumerate}

\bibliographystyle{IEEEtran}
\bibliography{VSKG_ref}

\end{document}